\documentclass[conference,a4paper]{IEEEtran}
\usepackage[utf8]{inputenc} 
\usepackage[T1]{fontenc}

\usepackage{ifpdf}

\usepackage{cite}

\ifCLASSINFOpdf
  \usepackage[pdftex]{graphicx}
\else
\fi
\usepackage[cmex10]{amsmath}
\usepackage{amsfonts}
\usepackage{amssymb}
\usepackage{amsthm}
\usepackage{color}
\usepackage{algorithmic}
\usepackage{array}
\usepackage{url}
\newtheorem{theorem}{Theorem}
\newtheorem{definition}{Definition}
\newtheorem{example}{Example}
\newtheorem{solution}{Solution}

\newtheorem{corollary}{Corollary}
\newtheorem{lemma}{Lemma}
\newtheorem{remark}{Remark}

\newcommand{\bx}{\mathbf{x}}
\newcommand{\cA}{\mathcal{A}}

\newcommand{\cP}{\mathcal{P}}
\newcommand{\cQ}{\mathcal{Q}}

\newcommand{\cB}{\mathcal{B}}
\newcommand{\cR}{\mathcal{R}}

\newcommand{\bV}{\mathbf{V}}
\newcommand{\bX}{\mathbf{X}}
\newcommand{\bo}[1]{\mathbf{#1}}
\newcommand{\ca}[1]{\mathcal{#1}}

\begin{document}
%
\title{Communication Complexity of $\min(X_1,X_2)$ with an Application to the Nearest Lattice Point Problem}
\author{\IEEEauthorblockN{Vinay A. Vaishampayan}
\IEEEauthorblockA{Dept. of Engineering Science and Physics\\City University of New York-College of Staten Island\\Staten Island, NY USA}
}
\maketitle

\begin{abstract}
Upper bounds on the communication complexity of finding the nearest lattice point in a given lattice $\Lambda \subset \mathbb{R}^2$ was considered in earlier works~\cite{VB:2017}, for a two party, interactive communication model. Here we derive a lower bound on the communication complexity of a key step in that procedure. Specifically, the problem considered is that of interactively finding $\min(X_1,X_2)$, when $(X_1,X_2)$ is uniformly distributed on the unit square.
A lower bound is derived on the single-shot interactive communication complexity and shown to be tight. This is accomplished by characterizing the  constraints placed on the partition generated by an interactive code and exploiting a self similarity property of an optimal solution. 
\end{abstract}



%
\IEEEpeerreviewmaketitle
\section{Introduction}

The \emph{communication complexity} (CC) of function computation is the minimum amount of information that must be communicated in order to compute a function of several variables with the underlying model that each variable is available to a distinct party~~\cite{Yao:1979},~\cite{KushNis:1997}. In a typical two-party setup~\cite{AhlCai:1994}, given alphabets $\ca{X}_1$,  $\ca{X}_2$ and $\ca{Z}$, and function $f~:~\ca{X}_1 \times \ca{X}_2 \rightarrow \ca{Z}$,   with each party having access to a block of observations represented as row vectors $\bx_1=(x_{11},x_{12},\ldots,x_{1k})$ and $\bx_2=(x_{21},x_{22},\ldots,x_{2k})$, respectively,  the objective is to determine the minimum amount of communication required so that each party can determine $(f(x_{11},x_{21}),f(x_{12},x_{22}),\ldots,f(x_{1k},x_{2k}))$ without error.  The case $k=1$ is referred to as the single-shot case, in which the objective is determine the minimum communication required to compute $f(x_1,x_2)$ (we drop the second subscript when $k=1$). Typical information theoretical results are obtained in the limit as $k\rightarrow \infty$. Let $\bx=[\bx_1,\bx_2]$ denote the $2\times k$ matrix with $i$th row $\bx_i$.

Given a lattice~\footnote{A lattice is a discrete additive subgroup of $\mathbb{R}^n$. The reader is referred to \cite{SPLAG} for details.} $\Lambda \subset \mathbb{R}^n$, the closest lattice point problem is to find for each $\bx=[x_1,x_2,\ldots,x_n]\in \mathbb{R}^n$, the point $\lambda_v(\bx)$ which minimizes
the Euclidean distance $\|x-\lambda\|$, $\lambda \in \Lambda$.  The Voronoi partition is the partition of $\mathbb{R}^n$ created by mapping $\bx$ to $\lambda_v(\bx)$. 

An upper bound on the CC of an approximate  nearest lattice point problem  was derived in~\cite{BVC:2017}, and an upper bound on the CC of 
transforming a nearest-plane or Babai partition to the Voronoi partition of $\mathbb{R}^2$ was derived in~\cite{VB:2017} both for the two-party single-shot  case ($\bx=[x_1,x_2]$). 
 An important step in that upper bound required the solution of the following problem: Two independent random variables, $X_1$ and $X_2$ have uniform marginal distributions on the unit interval $[0,1]$. How many bits must be exchanged on average in order to determine whether $X_1 < X_2$, or otherwise. In~\cite{VB:2017} an algorithm was presented that solved this using $R=4$ bits. Here, we show that this is optimal by deriving a lower bound on the amount of communication required. Bounds of this kind are referred to as single-shot converses in the information theory and computer science literature. 

The remainder of the paper is organized as follows. A brief review of relevant literature is in Sec.~\ref{sec:previous}, results needed for this paper from \cite{VB:2017} are in Sec.~\ref{sec:previous},  the entropy of the partition created by an infinite round algorithm is presented in Sec.~\ref{sec:infent}.  The main result, the single shot-converse is derived in Sec.~\ref{sec:converse}. Summary and conclusions are in Sec.~\ref{sec:summary}

\section{Previous Work}
\label{sec:previous}
Early information theoretic work on communication complexity for distributed function computation includes~\cite{Yam:1982},~\cite{AhlCai:1994}. Communication complexity for interactive communication is considered for worst case in~\cite{Orlitsky:1990}  and average case in \cite{Orlitsky:1992} where bounds on the communication rate are obtained in terms of specific graphs associated with the joint distribution~\cite{Wit:1976}.  A recent contribution shows the strict benefit of interactive communication for computing the Boolean AND function~\cite{MaIshwar:2009},~\cite{MaIshwar:2011}. A review of interactive communication and a discussion of open problems is in~\cite{Braverman:2012}. Most of the results obtained are for discrete alphabet sources. For continuous alphabet sources,  quantization for distributed function computation has been studied in~\cite{Misra:2011}. Converse results are rare in the quantization literature.    A recent converse result for entropy constrained scalar quantization is \cite{Koch:2016}. 

\section{The Bit-Exchange Protocol}
\label{sec:bitexchange}
Assume that the generator matrix $V$ of $\Lambda \subset \mathbb{R}^2$ has the upper triangular form $$\bV=\begin{pmatrix} 1 & \rho \cos \theta \\ 0 & \rho \sin \theta\end{pmatrix}$$ where the columns of $V$ are basis vectors for the lattice.  In~\cite{BVC:2017}, we computed an upper bound for the communication complexity of computing a Babai partition, which is a partition of $\mathbb{R}^2$ into rectangular Babai cells. A Babai cell tiles $\mathbb{R}^2$ , under the action of lattice translations, just as a Voronoi cell does. 

Here we briefly explain the construction in~\cite{VB:2017} for transforming a Babai partition to the Voronoi partition. At the start of the algorithm, both nodes know that $\bX=(X_1,X_2)$ lies in a Babai cell, which is the largest rectangular region in Fig.~\ref{fig:pi}. The objective is to assign some of the points in the Babai cell to neighboring lattice points, or equivalently to repartition the Babai cell with the boundaries of the Voronoi cell. This is accomplished by partitioning the Babai cell into seven rectangular sub-rectangles, three of which are error-free (they are entirely contained in the Voronoi cell for the origin) and four non-error-free rectangles (whose interior is  intersected by the boundary of the Voronoi cell). 

\begin{figure}[h!]
\begin{center}
		\includegraphics[height=3.7cm]{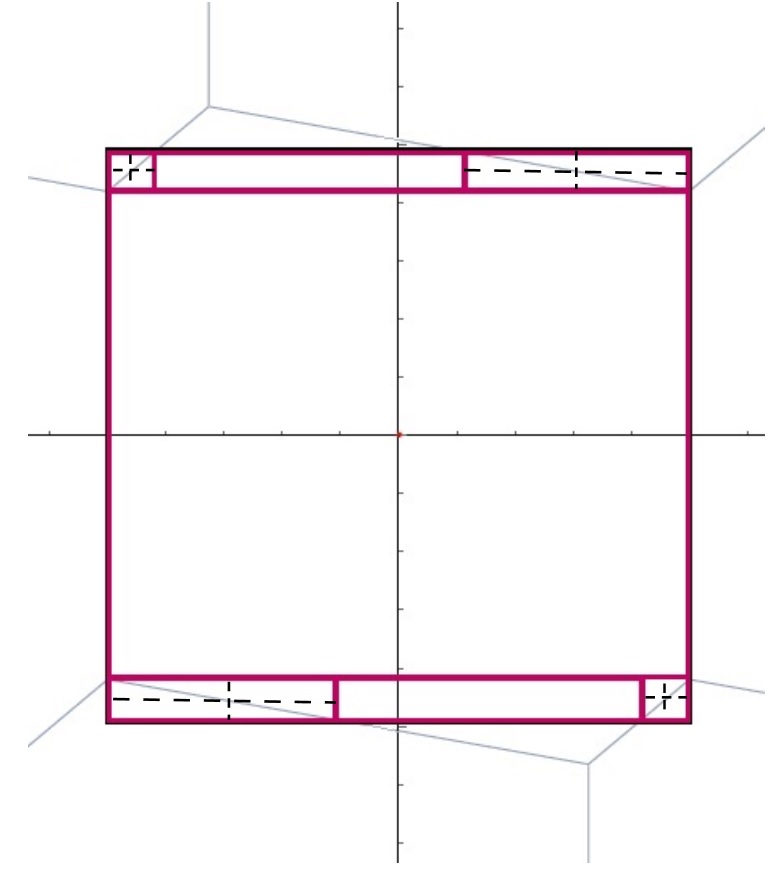}  
\caption{Voronoi partition (gray lines), one cell of the Babai partition (largest rectangle with red solid lines), rectangular partition of the Babai cell after the first round of communication (solid red lines) and rectangular partition after  the second round of communication (dashed lines). (from~\cite{VB:2017})}
\label{fig:pi}
\end{center}
\end{figure}

The cost of reconfiguring the partition with zero error is given by the following theorem~\cite{VB:2017}. A round refers to two messages, one from each node.
\begin{theorem}\cite{VB:2017}
For the interactive model with unlimited rounds of communication, a nearest plane partition can be transformed into the Voronoi partition using,  on average, a finite number of bits ($\bar{R}$) and rounds ($\bar{N}$) of communication. Specifically, 
\begin{equation}\bar{R}=H(Q)+(1-Q_0)H(P)+4(1-P_0)(1-Q_0) 
\label{eqn:avrateinf} 
\end{equation}
and
$$\bar{N}=1+2(1-P_0)(1-Q_0).$$
\end{theorem}
Here $P$, $Q$ are  probability distributions and $P_0$, $Q_0$ are probabilities determined by the shapes of the Voronoi and Babai cells---more details are in \cite{VB:2017}. The individual terms in (\ref{eqn:avrateinf}) are best explained by Fig.~\ref{fig:pi}. The first two terms come from the first round of communication. The last term comes from partitioning the four partition cells that cause errors (i.e. whose interiors intersect the Voronoi partition boundary). 
We  take a closer look at the last term in (\ref{eqn:avrateinf}), which is the cost of partitioning a rectangular region into two triangular regions.  This is the problem of finding the minimum of two independent random variables uniformly distributed on the unit interval. We refer to this problem  as the $\Pi_{\min}(2)$ problem.

Our construction  achieves an average cost of \emph{four} bits for constructing such a refinement. This is accomplished by constructing binary expansions for each $X_i$ (after a suitable shift and rescaling) and sequentially exchanging bits until the two bits differ. Thus if node-1 has bit string 00100001... and node-2 has bit string 00101001..., then five rounds of communication occur after which both nodes know that $X_1 < X_2$.  We  show that \emph{four} bits are optimal on average.

\section{Interactive Communication and Entropy of a Partition}
\label{sec:infent}
We now analyze the interactive model in which an infinite number of communication rounds are allowed. Here we explain the setup and prove a basic theorem regarding the sum rate of an interactive code. 

The setup for the interactive code is as follows. We are given two independent random variables $X_1$ and $X_2$ with known joint probability distribution and the objective is to compute $f(x_1,x_2)$ interactively. 
Let $U_0$ be a constant random variable.  Communication proceeds in rounds according to a pre-arranged protocol and each round consists of at most two steps (messages). In the $i$th step, $i \geq 1$, $i$ odd,  node $X_1$ sends message $U_i$ to node $X_2$ and for $i$ even, $i>1$, node $X_2$ sends message $U_i$ to $X_1$. For $i$ odd, $U_i$ depends on $X_1$ and $U^{i-1}:=(U_0,U_1,\ldots,U_{i-1})$ and for $i$ even, $U_i$ depends on $X_2$ and $U^{i-1}$, thus obeying the Markov conditions $U_i-(X_1,U^{i-1})-X_2$ for $i$ odd and $X_1-(U^{i-1},X_2)-U_i$ for $i$ even. The algorithm stops after concluding the $T$th step, if both nodes can determine $f(X_1,X_2)$, based on their private information ($X_1$ or $X_2$) and the communication transcript $U^T$. We can think of $T$ as a conditional stopping time relative to $(U_0,U_1,\dots)$ with side information $X_1$ and $X_2$ which obeys the two Markov conditions $(X_1,X_2)-(X_1,U^T)-T$ and $(X_1,X_2)-(X_2,U^T)-T$.

The sum rate, $R_{sum}$, is given by 
\begin{eqnarray}
R_{sum}= & \sum_{i=1,~i~odd}^TH(U_i|U^{i-1},X_2)+ \nonumber \\
&  ~~+\sum_{i=1,~i~even}^TH(U_i|U^{i-1},X_1).
\end{eqnarray}

The following theorem allows us to write
\begin{eqnarray}
R_{sum} & =  & \sum_{i=1, i~odd}^TH(U_i|U^{i-1}) + \nonumber \\
& = & ~~+\sum_{i=1,~i~even}^TH(U_i|U^{i-1}) \nonumber \\
& = & H(U^T,T).
\end{eqnarray}

%
\begin{theorem}
Let random variables $X_1$ and $X_2$ be independent and $U_i,~i=1,2,\ldots$ satisfy $U_i-(X_1,U^{i-1})-X_2$ for $i$ odd, and $U_i-(X_2,U^{i-1})-X_1$ for $i$ even. Let $U_0$ be a  constant. Then $X_1-U^{i}-X_2$ and  $$H(U_i|U^{i-1},X_2)=H(U_i|U^{i-1}),~~i~\mbox{odd}$$ and $$H(U_i|U^{i-1},X_1)=H(U_i|U^{i-1}),~~i~\mbox{even}.$$
\end{theorem}
\begin{proof}
We first prove that  $X_1-U^{i}-X_2$, by induction. Clearly this holds for $i=0$, since $X_1, X_2$ are independent. Assume that $X_1-U^{i}-X_2$ holds for $i$ even. Then for $i$ odd
\begin{eqnarray}
\lefteqn{H(X_1,X_2|U^{i})  = } &  & \nonumber \\
&= &  H(X_1|U^{i})+H(X_2|U^{i},X_1) \nonumber \\
 & \stackrel{(a)}{=} &  H(X_1|U^i)+H(X_2|U^{i-1},X_1) \nonumber \\
 &\stackrel{(b)}{=} &  H(X_1|U^i)+H(X_2|U^{i-1}) \nonumber \\
 & \geq &   H(X_1|U^i)+H(X_2|U^{i}).
\end{eqnarray}
where (a) is by hypothesis, (b) is by the induction hypothesis.  Equality follows because 
 the reverse inequality is always true. A similar argument holds for $i$ even.
 
Now consider the two identities in the theorem.  Let $i$  be odd. Then
\begin{eqnarray}
\lefteqn{H(X_2,X_1,U_i|U^{i-1}) = } & & \nonumber \\
& = & H(X_2|X_1,U_i,U^{i-1})+H(X_1,U_i|U^{i-1}) \nonumber \\
& = & H(X_2|X_1,U^{i-1})+H(X_1,U_i|U^{i-1})
\label{eqn-first}
\end{eqnarray}
where the final identity follows due to a hypothesis in the theorem statement.
However
\begin{eqnarray}
\lefteqn{H(X_2,X_1,U_i|U^{i-1}) = } & & \nonumber \\
& = &  H(X_2|U^{i-1})+H(X_1,U_i|U^{i-1}).
\label{eqn-second}
\end{eqnarray}
Comparing (\ref{eqn-first}) and (\ref{eqn-second}) we see that for $i$ odd, $X_2$ and $U_i$ are conditionally independent given $U^{i-1}$. A similar proof follows for $i$ even.  
\end{proof}

\begin{corollary}
Let $p_i,i=1,2,\ldots$ denote the probability of the $i$th cell of the  rectangular partition constructed by the algorithm, let $p=(p_1,p_2,\ldots)$ and let $H(p)$ be its entropy. Since each $p_i$ is associated with a unique realization $(u_1,u_2,\ldots,u_T)$, it follows that 
$R_{sum}=H(p)$.
\label{cor:one}
\end{corollary}

\begin{remark}
Each cell of the  partition of $[0,1]\times[0,1]$ constructed by the previously described model for  interactive communication    is a Cartesian product $\cA_1 \times \cA_2$, where $\cA_1$ and $\cA_2$ are subsets of $[0,1]$, which in general depend on the cell of the partition.
\end{remark}

\section{Single-Shot Converses }
\label{sec:converse}
The problem that we consider is as follows. Independent random variables $X_1$ and $X_2$ are uniformly distributed on the unit interval $(0,1)$, $X_1$ is observed at Node-1 and $X_2$ is observed at Node-2. The nodes exchange information in order to compute $f(\bo{X})$, where $\bo{X}=(X_1,X_2)$.   Communication between the nodes proceeds interactively using a predetermined strategy.

We begin by illustrating a converse with a simple example. 
\begin{example}
\label{ex:one}
Let $f$ be given by
\begin{equation}
f(\bx)=\left\{\begin{array}{cc}
1, & x_1 >1/2,~ x_2 > 1/2 \nonumber \\
0, & \mbox{otherwise}.
\end{array} \right.
\end{equation}
 as illustrated in Fig.~\ref{fig:example1}.
Find a lower bound on the  sum-rate for computing  $f$  interactively using an unbounded number of rounds of communication.  

\begin{figure}[h] 
   \centering
   \begin{tabular}{cc}
   \includegraphics[width=1.0in]{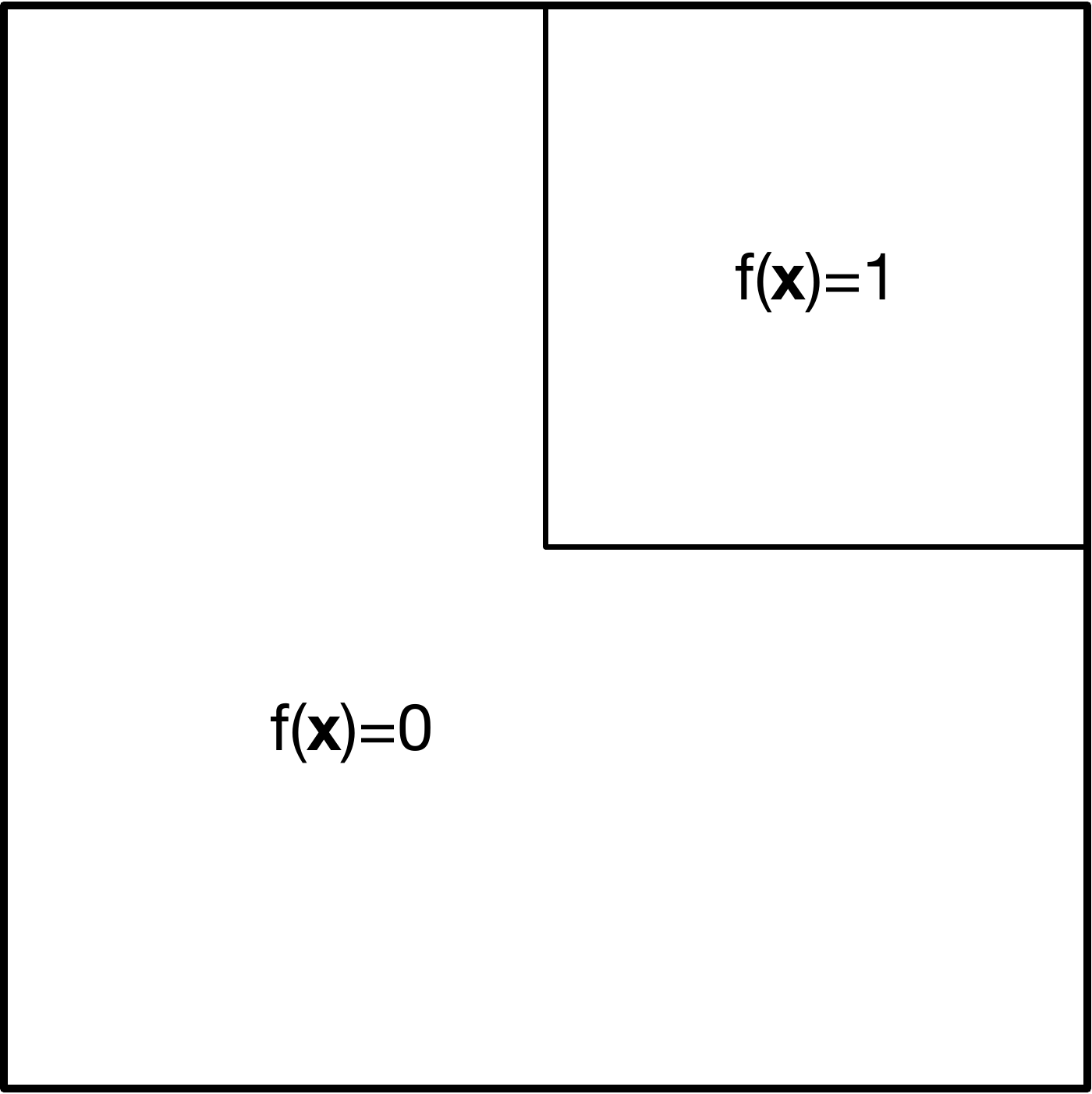} & \includegraphics[width=1.0in]{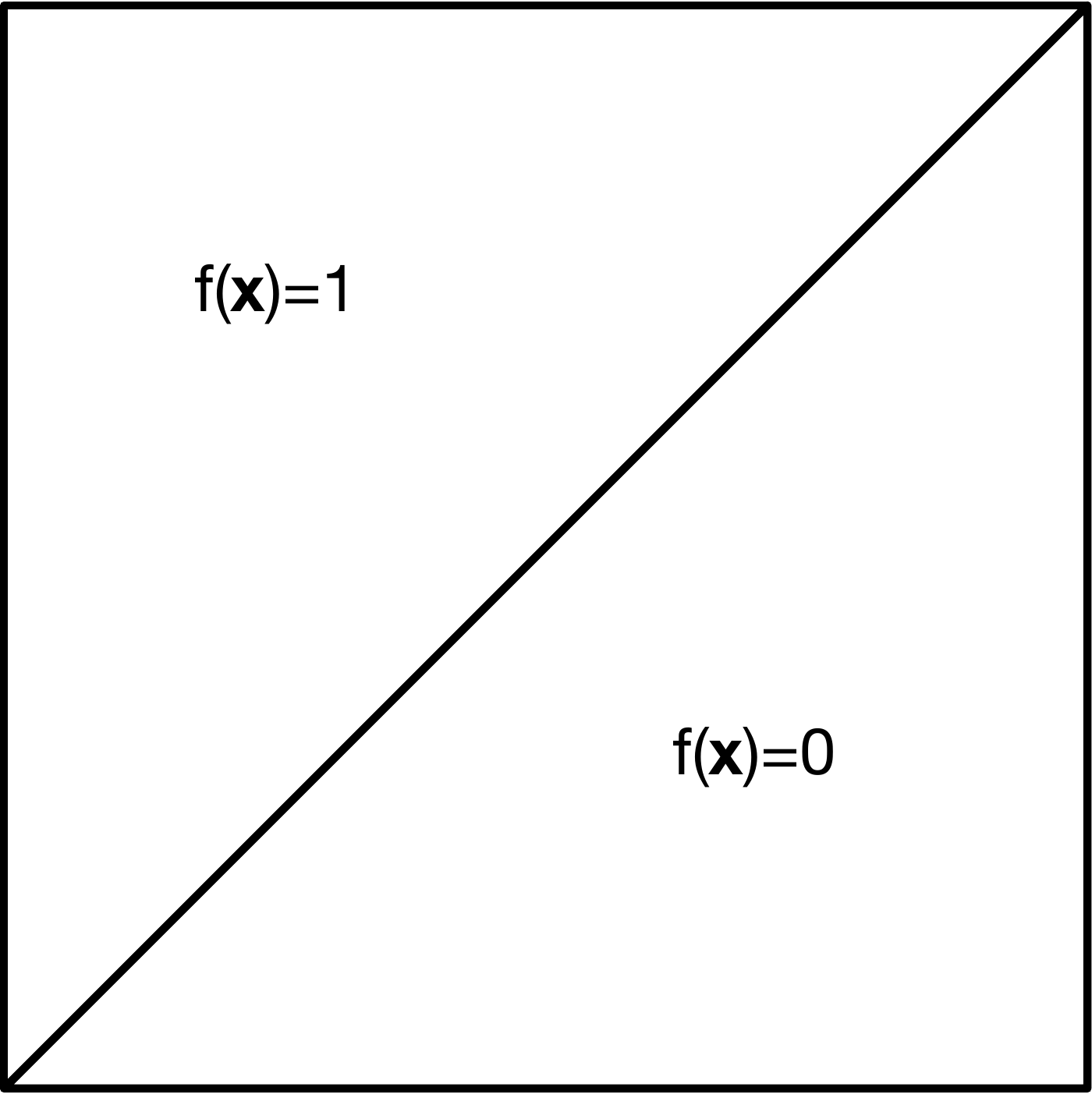}  \\
   (a) & (b) 
   \end{tabular}
   \caption{The unit square $[0,1]\times[0,1]$ and (a) the function $f$ used in Ex.~\ref{ex:one} (b) $f$  for $\Pi_{min}(2)$.}
   \label{fig:example1}
\end{figure}
\end{example}

\begin{solution}{\rm
An interactive algorithm creates a rectangular partition. We use the term partition  in the sense that the interior of the cells of the partition are assumed to be non-overlapping, while the union of the closure of the partition cells is the unit square $[0,1]^2$. The cells of any zero-error partition, i.e. a partition that achieves a zero probability of error fall into one of two sub-partitions, $\cP$ whose cells partition the set $\cR_p=\{\bx~:~f(\bx)=1\}$ and $\cQ$ whose cells partition the set $\cR_q=\{\bx~:~f(\bx)=0\}$.  Let $p_i$, $q_i$ denote the probability of the $i$th cell of $\cP$, $\cQ$, respectively. From the geometry of the problem, the constraint set $\Xi$ is defined by the following self-evident constraints: (i) $\sum_{j=1}^mp_{i_j} \leq 1/4$, $m=1,2,\ldots,$ for any subsequence $i_j$, and  (ii) $q_{i} \leq 1/2$, $i=1,2,\ldots$, $\sum_{j=1}^mq_{i_j} \leq 3/4$, $m=2,\ldots,$ for any increasing subsequence $i_j$. Thus any single probability $q_i$ cannot exceed $1/2$ and the sum of any pair of $q$ probabilities cannot exceed $3/4$.

From Cor.~\ref{cor:one}, the sum-rate is equal to the entropy of  the partition created by the algorithm, which in turn cannot be smaller than the infimum of the entropy of any partition that respects the probability constraints described above. Since the entropy function is a concave function of the probability distribution, and the constraint set is closed and bounded, the infimum is achieved at one of the corners of the constraint set. 

Consider probability row vector $(\bo{p},\bo{q})$, defined in terms of row vectors $\bo{p}=(p_1,p_2,\ldots)$ and $\bo{q}=(q_1,q_2,\ldots)$. Vertices of the convex constraint set $\Xi$ are $(\pi_1(\bo{p}^*), \pi_2(\bo{q}^*))$ with $\bo{p}^*=(1/4,0,0,\ldots)$ and $\bo{q}^*=(1/2,1/4,0,0,\ldots)$ and $\pi_1$ and $\pi_2$ are permutations of the coordinates of their vector arguments.

Clearly $\inf_{(\bo{p},\bo{q}) \in \Xi} H(\bo{p},\bo{q})=3/2$ bits. \hfill{$\square$}
}
\end{solution}

\begin{remark}
Since there is  also a simple algorithm for computing $f$ that requires a sum rate of $3/2$ bits, the lower bound on the sum-rate is tight. Also, the lower bound can be achieved using one round of communication.
\end{remark}
\begin{remark}
Optimization problems of the kind considered in the above example arise in facility placement problems and are classified as  \emph{geometric programming} problems~\cite{BW:2004}.
\end{remark}
\subsection{Lower Bound for $\Pi_{min}(2)$}


We now consider the problem that appears in the nearest lattice point problem, namely $\Pi_{min}(2)$ in which  we work with the function 
\begin{equation}
f(\bx)=\left\{\begin{array}{cc}
1, & x_1 \geq x_2 \nonumber \\
0, & \mbox{otherwise}.
\end{array} \right.
\end{equation}
Interactive communication results in a rectangular partition $(\cP,\cQ)$ where  $\cP$ is a subpartition of the region $\cR_p=\{\bx~:~f(\bx)=1\}$ and $\cQ$ of $\cR_p=\{\bx~:~f(\bx)=0\}$. Since the error probability is zero, we refer to $(\cP,\cQ)$ as a \emph{zero-error} partition.  The boundary is represented by the set $\cB:=\{x_1=x_2\}\bigcap (0,1)^2$.  For a zero-error partition $(\cP,\cQ)$ it is true that each point of $\cB$  must be the upper left corner of some rectangle that lies entirely in $\cR_p$ or the lower right corner of some rectangle that lies entirely in $\cR_q$, except possibly for a set of one-dimensional measure zero. Let $\{p_i,~i=1,2,\ldots\}$ be the probabilities of  the   cells of $\cP$   and let $\{q_i,~i=1,2,\ldots\}$ be the  probabilities of the  cells of $\cQ$. 


\begin{figure}[h!]
\centering
\begin{tabular}{cc}
		\includegraphics[height=2.7cm]{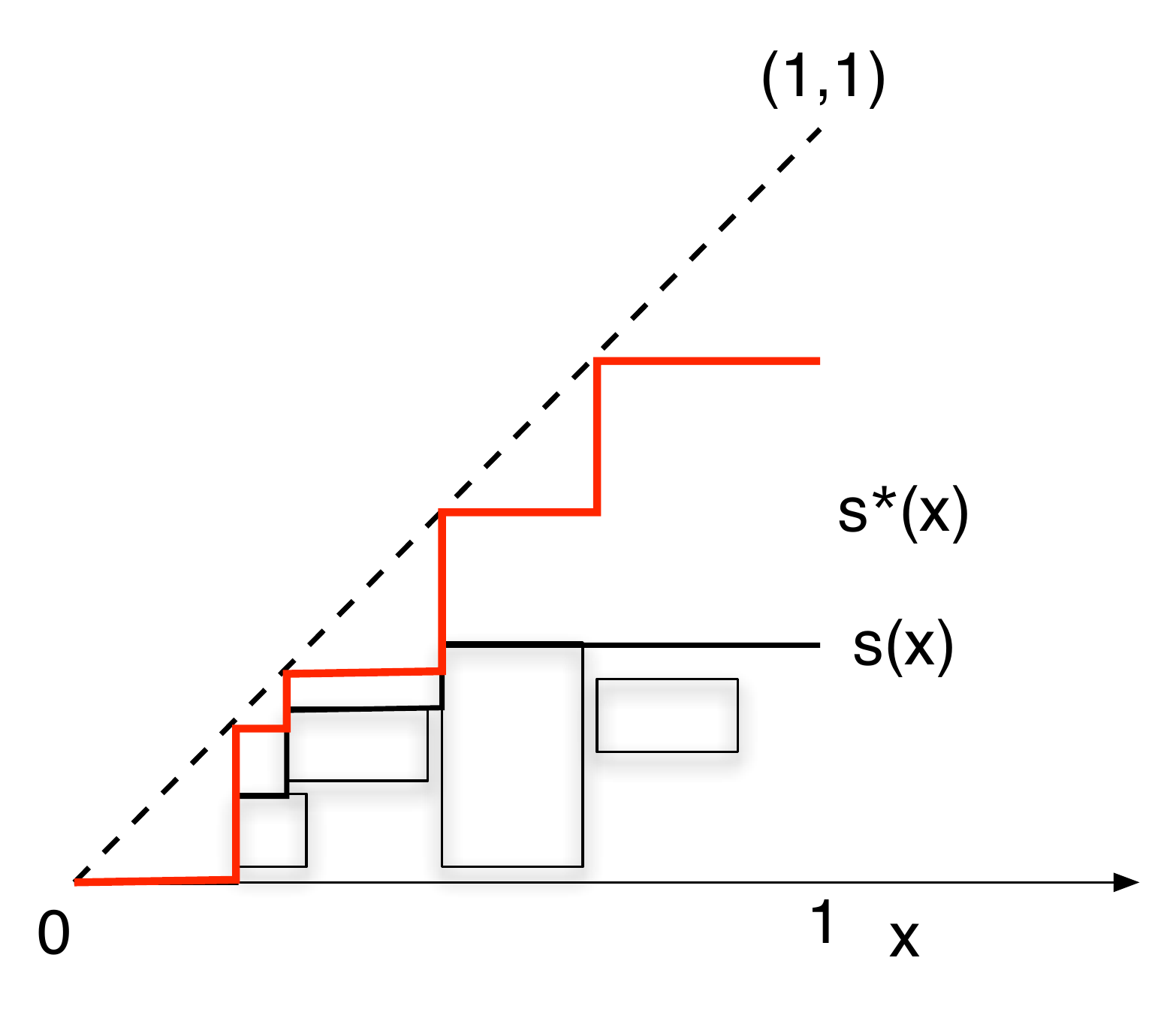}  &
		\includegraphics[height=2.7cm]{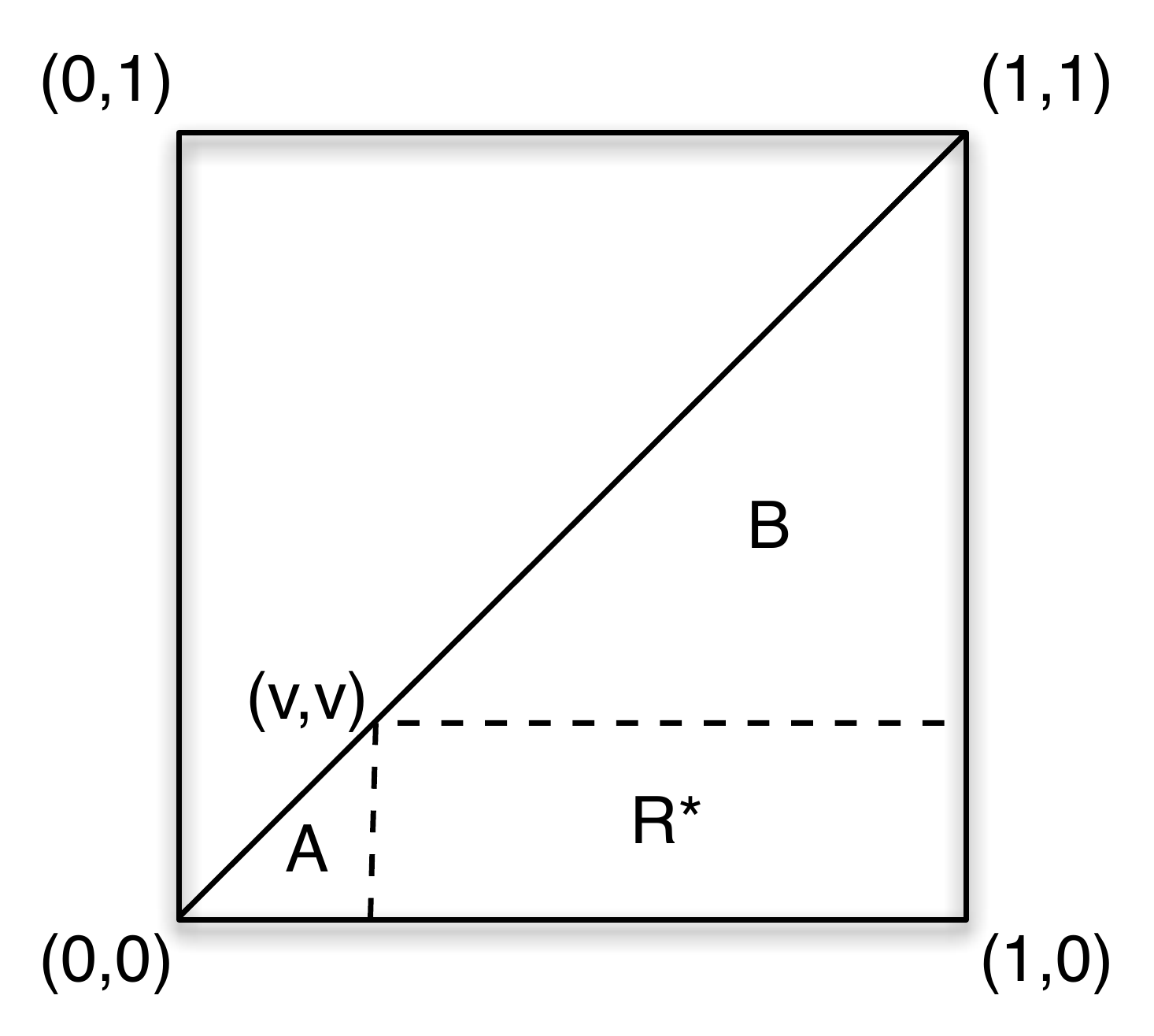}  \\
		(a) & (b)
		\end{tabular}
\caption{(a) Staircase functions used in the proof of Thm.~\ref{thm:constraints}, (b) Regions used in Thm.~\ref{thm:fourbits}. $R_p=A\bigcup R^* \bigcup B$.}
\label{fig:staircaseregions}
\end{figure}

\begin{theorem}
\label{thm:constraints}
The partition probabilities of a zero-error partition $(\cP,\cQ)$  satisfy the following constraints:
\begin{eqnarray}
\sum_{j=1}^m p_{i_j} & \leq  & \frac{m}{2(m+1)},~m=1,2,\ldots   \label{eqn:cp1}\\
\sum_{j=1}^m q_{i_j} & \leq & \frac{m}{2(m+1)},~m=1,2,\ldots  \label{eqn:cq1}\\
\sum_i p_i & = & 1/2,  \label{eqn:cp2} \\
\sum_i q_i  & = & 1/2, \label{eqn:cq2}
\end{eqnarray}
for any increasing subsequence of positive integers  $\{i_j\}$.
\end{theorem}

\begin{proof}
Consider any rectangular partition of $\cR_p$ and consider any $m$ cells of the partition. Construct a non-decreasing staircase function, $s(x_1)$, whose height at $x_1$ is the maximum $x_2$ coordinate of any of the  cells of $\cP$ with a vertex that lies to the left of $x_1$ (Fig.~\ref{fig:staircaseregions}). Such a staircase function is piecewise constant over at most $m+1$ intervals, and thus partitions the interval $[0,1]$ into at most $m+1$ cells. The area of the union of the selected partition cells is upper bounded by the area of a modified staircase function $s^*(x_1)$ obtained by pushing each horizontal segment of $s(\cdot)$ upwards until it touches the boundary $\cB$.
 Let $(x_i,x_i),~i=1,2,\ldots,m$, $x_1 \leq  x_2  \leq \ldots \leq  x_m$ be the top left corners of the rectangular cover. The area of the rectangular cover, i.e. the area under the modified staircase function  is given by $x_1(1-x_1) +(x_2-x_1)(1-x_2) +(x_3-x_2)(1-x_3)+\ldots+(x_m-x_{m-1})(1-x_m)$. It is easy to check that the area is a concave function  of $x_1,x_2,\ldots,x_m$, since the Hessian matrix, $H$,  a symmetric Toeplitz $m \times m$ matrix with top row $(-2,~1,~0,\ldots,0)$, is negative definite.  This can be checked directly from the real quadratic form associated with this matrix, $x^t H x=-x_1^2-\sum_{i=2}^m (x_{i-1}-x_i)^2 -x_m^2$. This function is maximized by setting $x_i=i/(m+1)$, $i=1,2,\ldots,m$ and the maximum area is given by $m/2(m+1)$.
\end{proof}

If a single partition cell has probability $1/4$ it must be a square whose top left corner is $(1/2,1/2)$, and similarly if $m$ cells meet the upper bound on the sum of their areas, then the top left corners of the staircase function $s(\cdot)$ must be at the points $(x,x)$ with $x \in \{1/(m+1),2/(m+2),\ldots,m/(m+1)\}$. 

The set of constraints in Thm~\ref{thm:constraints} defines a polyhedron of probability vectors $(p,q)$, which contains the set of probability vectors as constrained by the partition, but the inclusion is strict. As an example, consider the point $(1/4,1/12,1/24,...)$, an extreme point of the the set of inequalities (\ref{eqn:cp1}--\ref{eqn:cq2}). This point is not realizable by any partition of a triangle since as soon as $P_1=1/4$, $P_2$ cannot be larger than $1/16$. A sufficiently tight characterization of the probabilities that respect the partition appears to be rather complicated, and in our attempts did not lead to a useful conclusion. However, for this problem, majorization~\cite{HJ:1985}, plays a significant role. For convenience we state the definition.

\begin{definition}
Let $p=(p_1,p_2,\ldots,)$ and $q=(q_1,q_2,\ldots,)$ be two probability vectors with probabilities in nonincreasing order. Then $p$ majorizes $q$, written $p\succeq q$, if $\sum_{i=1}^k p_i \geq \sum_{i=1}^k q_i$, for $k=1,2,\ldots$.
\end{definition}

\begin{lemma}\cite{CoverThomas:2006}
If $p \succeq q$ then $H(p) \leq H(q)$.
\end{lemma}

\begin{figure}[htbp] 
   \centering
   \begin{tabular}{cc}
   \includegraphics[width=2.9cm]{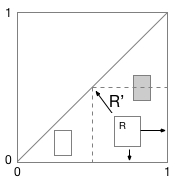} &
   		\includegraphics[height=2.9cm]{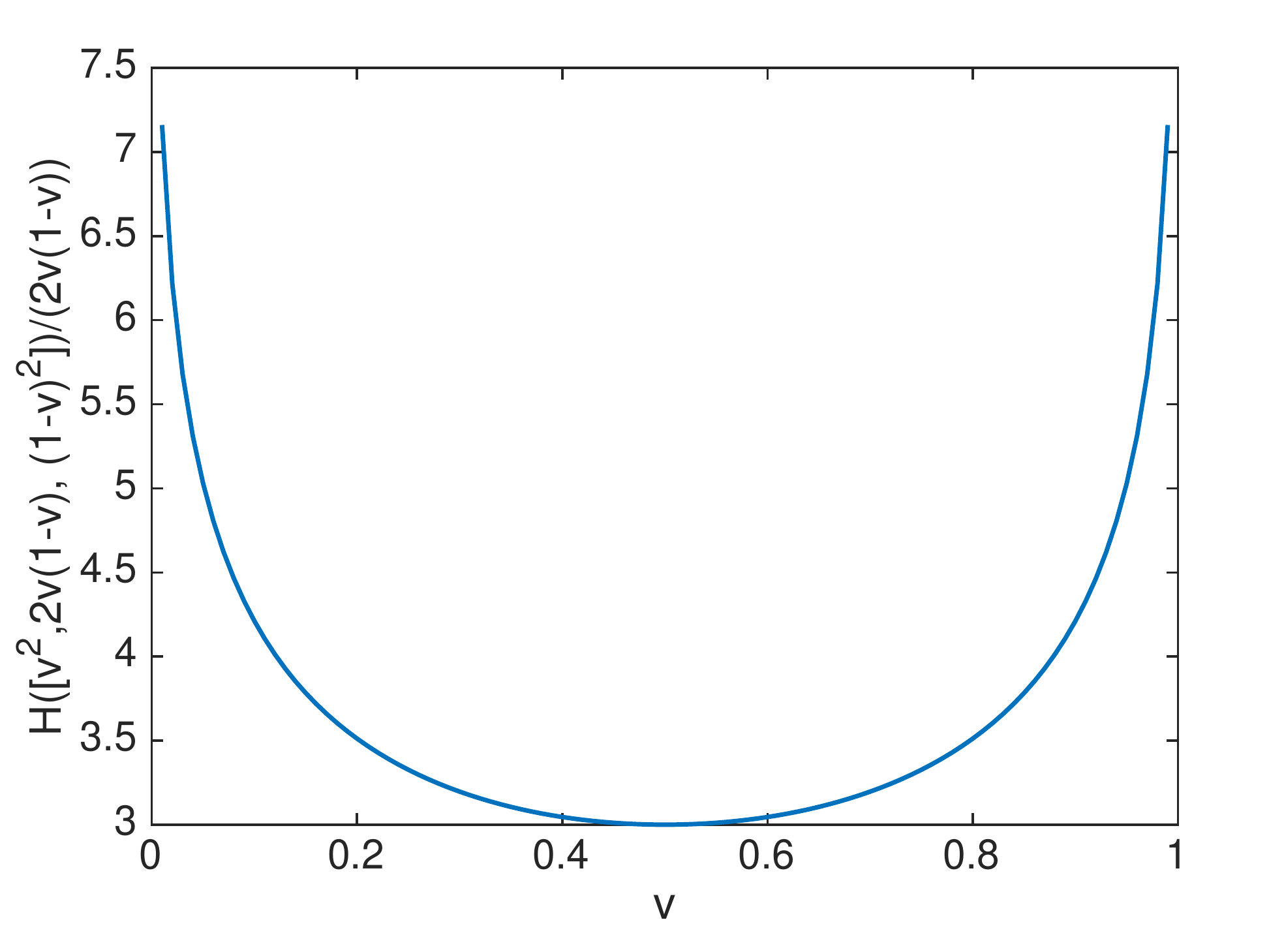}  
		\end{tabular} 

   \caption{(left) Illustration of the readjustment of the rectangle $R$ with the highest probability, (right) Plot of Eqn.~(\ref{eqn:entropyratio}).}
   \label{fig:Readjust}
\end{figure}

\begin{theorem}
\label{thm:pathconnected}
If a partition minimizes the entropy it contains a rectangle with vertices $(1,0)$ and   $(v,v)$ and  another rectangle with vertices $(0,1)$ and  $(u,u)$, for some $0<u,v <1$. 
\end{theorem}
\begin{proof}
Proof is by contradiction. Suppose we have a partition $\cP$ for the triangular region $\cR_p$, in which a rectangle $R$ with the largest probability does not have vertices $(u,u)$ and $(1,0)$. Construct a new partition from $\cP$ by moving the faces of the rectangle $R$ outwards, creating a new rectangle $R'$ which contains $R$, as illustrated in Fig.~\ref{fig:Readjust}(left). The process does not create any new cells, any cell that now lies in $R'$ has its probability reduced to zero, and any cell partially intersected by $R'$ has its probability reduced to the part outside $R'$. Thus, if a cell other than $R$ is affected, say cell $j$, then its probability $P'_j=(1-\alpha)P_j$,  where $0 < \alpha < 1$. Also,  $\alpha P_j$  is added to $P_1$. Let $\mathcal J$ denote the set of affected cells (other than $R$). Then $P'_1=P_1+\sum_{j \in \mathcal J}\alpha_jP_j$  and $P'_j=(1-\alpha_j)P_j$, $j \in \mathcal J$. Suppose the probabilities of $\cP$ arranged in nonincreasing order are $p=(p_1,p_2,\ldots)$. Let $p'=(p_1',p_2',\ldots)$ denote the probabilities of the new partition and let $p''$ be obtained by sorting $p'$ in nonincreasing order. Let $[k]=\{1,2,\ldots,k\}$. Then $p''_1=p'_1 \geq p_1$ and $\sum_{i=1}^k p''_i \geq \sum_{i=1}^k p'_i = p_1+\sum_{j \in \mathcal J}\alpha_jp_j + \sum_{j \in [k]\bigcap {\mathcal J}^c}p_j +\sum_{j \in [k]\bigcap {\mathcal J}}(1-\alpha_j)p_j=\sum_{j=1}^k p_j +\sum_{j \in [k]^c \bigcap \mathcal J}(1-\alpha_j)p_j \geq \sum_{j=1}^k p_j$. Thus $p'' \succeq p$ and $H(p'') \leq H(p)$.
\end{proof}

\begin{theorem}
The minimum single-shot interactive communication cost of the $\Pi_{min_2}$ problem is \emph{four} bits.
\label{thm:fourbits}
\end{theorem}
\begin{proof}
Consider an extreme partition $(\cP,\cQ)$ which contains a rectangle $R^*$ which has a points $(v,v)$ and $(1,0)$ as its upper left and lower right vertices.  This is always true by Thm.~\ref{thm:pathconnected}. Let random variable $C$ indicate whether $(x_1,x_2)$ lies in $\cR_p$ or not, and let random variable $S$ indicate whether $(x_1,x_2)$ lies in one of the three regions, $R^*$, $A$ or $B$ as shown in Fig.~\ref{fig:staircaseregions}(b). Let $H(\cP,\cQ))$ denote the entropy of the partition $(\cP,\cQ)$. Then 
\begin{eqnarray}
H(\cP,\cQ) & = & H(C)+H(\cP,\cQ|C=0)P(C=0)+ \nonumber \\
& & ~~H(\cP,\cQ|C=1)P(C=1) 
\label{eqn:firstdecomp}
\end{eqnarray}
and
\begin{eqnarray}
\lefteqn{H(\cP,\cQ|C=1)= H(S|C=1)+} & \nonumber \\ 
& +H(\cP,\cQ|C=1,S=A)P(S=A|C=1)+ \nonumber \\
& + H(\cP,\cQ|C=1,S=B)P(S=B|C=1). 
\end{eqnarray}
Since the regions $A$ and $B$ are similar to $\cR_p$ it follows that if this partition minimizes the entropy it must satisfy the recursion
\begin{eqnarray}
\lefteqn{H(\cP,\cQ|C=1)  =  H([v^2,2v(1-v),(1-v)^2])+} & & \nonumber \\
&  & ~~~~~+ (v^2+(1-v)^2)H(\cP,\cQ|C=1).
\end{eqnarray}
Solving for $H(\cP,\cQ|C=1)$ we obtain
\begin{equation}
H(\cP,\cQ|C=1)=\frac{H([v^2,2v(1-v),(1-v)^2])}{2v(1-v)}
\label{eqn:entropyratio}
\end{equation}
whose unique minimum value of $3$ bits occurs when $u=v=1/2$ (see Fig.~\ref{fig:Readjust}). Plugging back in (\ref{eqn:firstdecomp}) leads to the desired result.
\end{proof}

\section{Summary and Conclusions}
\label{sec:summary}
A lower bound on the communication complexity of $f(x_1,x_2)=\min(x_1,x_2)$ has been derived for the two-party, single shot case, interactive case with an unbounded number of rounds of communication, when $\bx=[x_1,x_2]$ is uniformly distributed on the unit square. This lower bound has been derived by showing that the amount of communication required is equal to the entropy of the partition created by the communication between the two parties, and then deriving a lower bound on the entropy of  rectangular partitions that are constrained by the geometry of the problem. The problem is shown to be related to geometric programming problems encountered in optimizing facility locations.

\end{document}